\title{Counting Dominating Sets in Directed Path Graphs}
\author{
        \textsc{Min-Sheng Lin}\thanks{Email: mslin@ntut.edu.tw}\\ %
        Department of Electrical Engineering\\
        National Taipei University of Technology\\
        Taipei 106, Taiwan, ROC\\
}
\documentclass[11pt]{article}
\usepackage[paper=a4paper,dvips,top=1.5cm,left=2cm,right=2cm, foot=1cm,bottom=1.5cm]{geometry}
\usepackage{titlesec}
\usepackage{amsmath}
\usepackage{amsthm}
\usepackage{enumerate}
\usepackage[ruled]{algorithm2e}
\usepackage{graphicx}
\usepackage{cite}

\titleformat*{\section}{\large\bfseries}
\titleformat*{\subsection}{\large\bfseries}
\titleformat*{\subsubsection}{\itshape\subsubsectionfont}
\newcommand{\keywords}{\textbf{ \textit{Keywords --- }}}
\newtheorem{thm}{Theorem}
 
\newtheorem{lem}{Lemma}
\newtheorem{cor}{Corollary}
\newtheorem{pty}{Property}

\begin{document}
\maketitle

\begin{abstract}
A dominating set of a graph is a set of vertices such that every vertex not in the set has at least one neighbor in the set. The problem of counting dominating sets is \#P-complete for chordal graphs but solvable in polynomial time for its subclass of interval graphs. The complexity status of the corresponding problem is still undetermined for directed path graphs, which are a well-known class of graphs that falls between chordal graphs and interval graphs. This paper reveals that the problem of counting dominating sets remains \#P-complete for directed path graphs but a stricter constraint to rooted directed path graphs admits a polynomial-time solution.
\end{abstract}

\section{Introduction}
For a graph $G$, a subset $S$ of vertices of $G$ is a \textit{dominating set} (DS) if every vertex of $G$ not in $S$ is adjacent to a vertex in $S$. This paper concerns the problem of computing the number of DSs in a graph. The problem of counting dominating sets (abbr. \#DS problem) is known to be \#P-complete for planar graphs  \cite{hunt1998complexity}, chordal graphs, chordal bipartite graphs \cite{kijima2011dominating}, and tree convex bipartite graphs \cite{lin2022counting}. Valiant \cite{valiant1979complexity} defined the class of \#P problems as those that involve counting access computations for problems in NP, and the class of \#P-complete problems includes the hardest problems in \#P. As is well known, all algorithms that exactly solve these problems have exponential time complexity, so efficient algorithms for solving this class of problems are unlikely to be developed. However, some polynomial-time algorithms for solving problem \#DS in interval graphs, trapezoid graphs \cite{kijima2011dominating}, and rooted path-tree bipartite graphs \cite{lin2022counting} have been found.

One very important class of graphs is the class of intersection graphs. Let $F$ be a finite family of non-empty sets. A graph $G$ is an intersection graph for $F$ if an isomorphism exists between the vertices of $G$ and the sets of $F$ such that two vertices are adjacent if and only if their corresponding sets in $F$ have a non-empty intersection. Examples of such graph classes are chordal graphs and interval graphs.

\textit{Chordal graphs} are graphs in which every cycle with a length of at least four has a chord. Gavril \cite{gavril1974intersection} proved that chordal graphs are the intersection graphs of a family of subtrees in a clique tree. A tree $T$ is a \textit{clique tree} for a graph $G$ if each node in $T$ corresponds to a maximal clique in $G$. To avoid confusion with the vertices of graph $G$, the vertices of tree $T$ are called nodes or clique nodes. Let $T_v$ denote the set of all cliques of $G$ that contain vertex $v$. $G$ is a chordal graph if and only if $T_v$ is a subtree in a clique tree $T$ for every vertex $v$ of $G$. The concept of chordal graphs suggests definitions of some subclasses of chordal graphs. \textit{Undirected path graphs} are the intersection graphs of a family of undirected subpaths in a clique tree. \textit{Directed path graphs} are the intersection graphs of a family of directed subpaths in a directed clique tree. \textit{Rooted directed path graphs} are the intersection graphs of a family of directed subpaths in a rooted directed clique tree. A tree is called a \textit{rooted directed tree} if one node has been designated as the root, and the edges have a natural orientation away from the root. \textit{Interval graphs} are rooted directed path graphs in which the clique tree is itself a path. These graph classes are related by the following proper inclusions; interval graphs $\subset$ rooted directed path graphs $\subset$ directed path graphs $\subset$ undirected path graphs $\subset$ chordal graphs \cite{monma1986intersection}.

The \#DS problem remains \#P-complete even for chordal graphs, but a polynomial-time algorithm exists for solving the problem for interval graphs \cite{kijima2011dominating}. It is known that one can count minimal dominating sets in strongly chordal graphs in polynomial time \cite{bergougnoux2019counting} and, thus, for rooted directed path graphs which form a proper subclass of strongly chordal graphs \cite{brandstadt2010rooted}. However, there is still no direct approach to count all dominating sets for rooted directed path graphs. The status of the \#DS problem for undirected path graphs, directed path graphs, and rooted directed path graphs has been undetermined until now. This paper reveals that the \#DS problem remains \#P-complete even when restricted to directed path graphs and undirected path graphs but that a stricter restriction to rooted directed path graphs admits a polynomial-time solution. Figure 1 summarizes the situation. Accordingly, the borderline between polynomial and \#P-complete problems is fully determined for some subclasses of chordal graphs in Fig. 1.

Some notation and terminology are introduced for later use. The sets of vertices and edges of a graph $G$ are denoted as $V(G)$ and $E(G)$, respectively. The neighborhood of a vertex $v$ in a graph $G$, denoted by $N_G(v)$, is the set of vertices that are adjacent to $v$ in $G$. A vertex $v$ is said to \textit{dominate} vertex $u$ if either $v = u$ or $v$ is adjacent to $u$; that is, if $v {\in} N_G(u) \cup {u}$. For subsets $S, U \subseteq V(G)$, $S$ is said to dominate $U$ in $G$ if for every vertex $u {\in} U$, there exists a vertex $v {\in} S$ such that $v$ dominates $u$. Thus, a subset $S \subseteq V(G)$ is a dominating set in $G$ if and only if $S$ dominates $V(G)$. Let $\#DS(G)$ denote the number of dominating sets in $G$.
\begin{figure}[htbp]
    \makebox[\textwidth]{\includegraphics[scale=1.0, clip, trim=4cm 19.7cm 4cm 2.7cm]{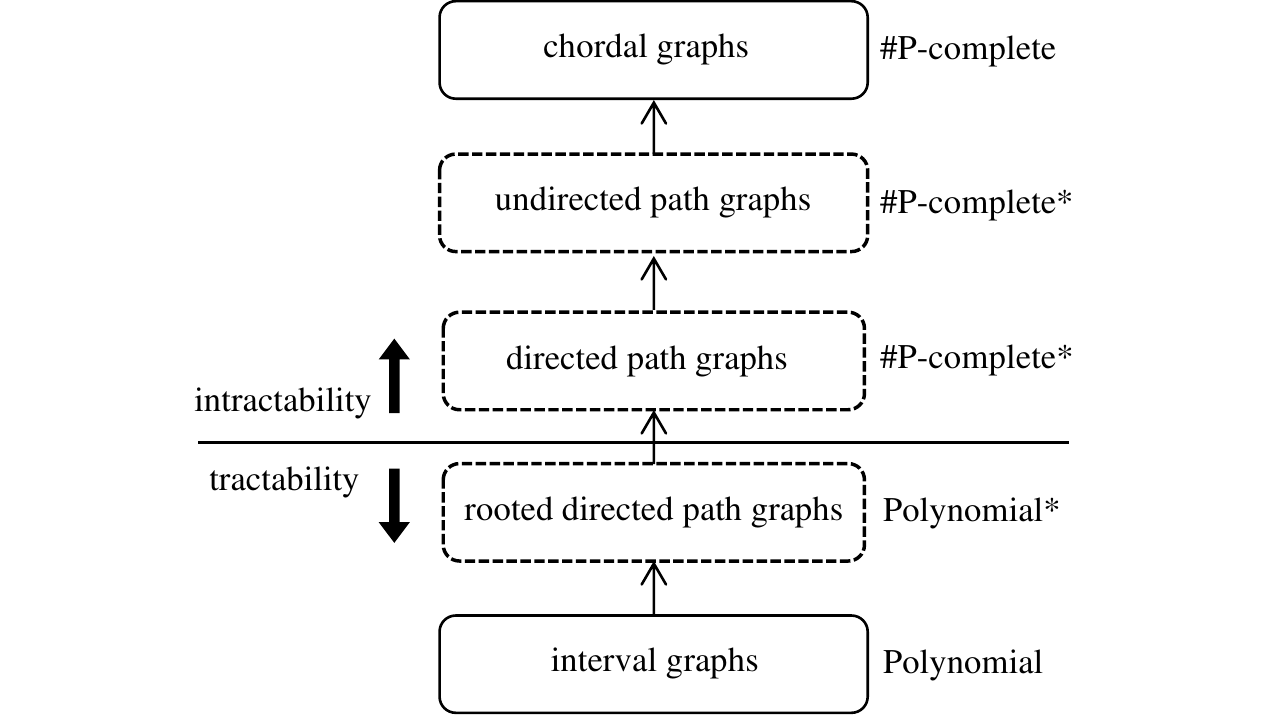}}
    \caption{ Status of \#DS problem for some subclasses of chordal graphs. A $\rightarrow$ B means that A is a subclass of B. Symbol * indicates a main contribution of this paper.}
    \label{fig1}
\end{figure}
\section{\#P-Completeness of the \#DS problem for Directed Path Graphs}
The \#DS problem for split graphs and chordal graphs has been proved to be \#P-complete by reduction from the problem of counting independent sets of bipartite graphs \cite{kijima2011dominating}. A slight variation of that reduction suffices to prove that the \#DS problem remains \#P-complete even when restricted to directed path graphs. 
\begin{thm}
The \#DS problem for directed path graphs is \#P-complete.
\end{thm}
\begin{proof}
The reduction proceeds from the problem of counting edge covers in a bipartite graph, which was proved to be \#P-complete by Provan and Ball \cite{provan1983complexity}.  An edge cover is a subset of edges that covers all the vertices. Let $B$ be an arbitrary bipartite graph with vertex bipartition $X={x_1, x_2, …, x_n}$ and $Y={y_1, y_2, …, y_m}$, and let $N=n+m$ denote the number of vertices of $B$. In the following steps, $N+1$ corresponding clique trees $T^r$, for $1{\leq} r {\leq} N+1$, are constructed from $B$, from which $N+1$ directed path graphs $G^r$ are obtained. The idea behind this construction is similar to the reduction that was used in \cite{lin2014counting} to show the \#P-completeness of the problem of counting maximal independent sets in directed path graphs. First, corresponding to each vertex $x_i {\in} X$, construct $r$ copies of clique node $K^s_i= \{x^s_i\} \cup \{e_{ij} : (x_i,y_j) {\in} E(B)\}$, for $1 {\leq} s {\leq} r$. Next, corresponding to each vertex $y_j {\in} Y$, construct $r$ copies of clique node $H^s_j= \{y^s_j\} \cup \{e_{ij} : (x_i, y_j) {\in} E(B)\}$, for $1 {\leq} s {\leq} r$. Finally, construct one large clique node $Q=\{e_{ij} : (x_i, y_j) {\in} E(B)\}$. Let $T^r$ be the directed clique tree such that $V(T^r)=\{Q\} \cup \{K^s_i : 1 {\leq} s {\leq} r \text{ and } 1 {\leq} i {\leq} n\} \cup \{ H^s_j : 1 {\leq} s {\leq} r \text{ and } 1 {\leq} j {\leq} m\}$ is the set of clique nodes of $T^r$; $E(T^r)=\{K^r_i {\rightarrow} K^{r-1}_i : 1 {\leq} i {\leq} n \} \cup \{K^{r-1}_i {\rightarrow} K^{r-2}_i : 1 {\leq} i {\leq} n\} \cup ... \cup \{K^2_i {\rightarrow} K^1_i : 1 {\leq} i {\leq} n\} \cup \{K^1_i {\rightarrow} Q : 1 {\leq} i {\leq} n\} \cup \{Q {\rightarrow} H^1_j : 1 {\leq} j {\leq} m\} \cup \{H^1_j {\rightarrow} H^2_j : 1 {\leq} j {\leq} m\} \cup ... \cup \{H^{r-1}_j {\rightarrow} H^r_j : 1 {\leq} j {\leq} m\}$ is the set of directed edges of $T^r$. Figure 2 presents an example of the above construction.

Let $G^r$ be the graph whose maximal cliques are the set $V(T^r)$. The resulting graph $G^r$, corresponding to the clique tree $T^r$, will now be shown to be a directed path graph. Let $X^r=\{ x^s_i  : 1 {\leq} s {\leq} r \text{ and } 1 {\leq} i {\leq} n \}$ and $Y^r=\{ y^s_j : 1 {\leq} s {\leq} r \text{ and } 1 {\leq} j {\leq} m \}$. Thus, $V(G^r)= X^r \cup Y^r \cup Q$. For $v {\in} V(G^r)$, let $P_v$ be the set of all clique nodes of $T^r$ that contain vertex $v$. If $v=x^s_i$ (or $v=y^s_j$) is a vertex of $X^r$ (or $Y^r$, respectively) then $P_v$ comprises the single clique node $K^s_i$ (or $H^s_j$, respectively) and $P_v$ is a directed path of length zero in $T^r$. If $v=e_{ij}$ is a vertex of $Q$, then $P_v$ comprises the directed path $K^r_i {\rightarrow} K^{r-1}_i {\rightarrow} ...{\rightarrow} K^1_i {\rightarrow} Q {\rightarrow} H^1_j {\rightarrow} H^2_j {\rightarrow} ... {\rightarrow} H^r_j$ of length $2r$ in $T^r$. Therefore, for each vertex $v {\in} V(G^r)$, $P_v$ is a directed path in $T^r$. Therefore, $G^r$ is a directed path graph with the corresponding clique tree $T^r$.

Now, the relationship between the number of DSs in $G^r$, for $1 {\leq} r {\leq} N+1$, and the number of edge covers in $B$ will be established. Clearly, every $G^r$, for $1 {\leq} r {\leq} N+1$, is also a split graph with a clique $Q$ and an independent set $X^r \cup Y^r$. Using the same strategy and techniques as were used in \cite{kijima2011dominating}, the number of DSs in $G^1$ is $\#DS(G^1) = \sum_{S \subseteq Q}2^{\left|N_{G^1}(S) \cap (X^1  \cup Y^1)\right|}$. Let $ z_k = \left| \{S \subseteq Q : \left|N_{G^1}(S) \cap (X^1  \cup Y^1)\right|=k\} \right| $. Thus,
\begin{equation} 
 \#DS(G^1) = \sum_{k=0}^N z_k \cdot 2^k.
\end{equation}
Equation (1) can be generalized as 
\begin{equation} 
\#DS(G^r) = \sum_{k=0}^N z_k \cdot (2^k)^r = \sum_{k=0}^N z_k \cdot (2^r)^k,  \text{ for } 1 \leq r \leq N+1.
\end{equation}
Additionally, Equation (2) can be expressed in matrix form as 
\begin{equation}
\mathcal{M} \times 
\begin{pmatrix}
z_0 \\
z_1 \\
z_2 \\
\vdots \\
z_N
\end{pmatrix}
=
\begin{pmatrix}
\#DS(G^1) \\
\#DS(G^2) \\
\#DS(G^3) \\
\vdots \\
\#DS(G^{N+1})
\end{pmatrix},
\end{equation}
where
\begin{equation}
 \mathcal{M} = 
 \begin{pmatrix}
 1 & \alpha_1^{\;1} & \alpha_1^{\;2} & \cdots & \alpha_1^{\,N} \\
 1 & \alpha_2^{\;1} & \alpha_2^{\;2} & \cdots & \alpha_2^{\,N} \\
 1 & \alpha_3^{\;1} & \alpha_3^{\;2} & \cdots & \alpha_3^{\,N} \\
 \vdots & \vdots  & \vdots & \ddots & \vdots \\
 1 & \alpha_{N+1}^{\;1} & \alpha_{N+1}^{\;2} & \cdots & \alpha_{N+1}^{\,N}
 \end{pmatrix}
 \text{ and } \alpha_r=2^r \text{ for } 1 \leq r \leq N+1.
 \end{equation}
$\mathcal{M}$ is a well-known Vandermonde matrix and all $\alpha_r$, for $1{\leq}r{\leq}N+1$, are distinct. Therefore, the determinant of $\mathcal{M}$ is non-zero. By Cramer's rule, given the values for $\#DS(G^1), \#DS(G^2), ..., \#DS(G^{N+1})$, the values of $z_0$, $z_1$, ..., and $z_N$ can be obtained in time that is polynomial in $N$. Notably, an edge cover of a graph is a set of edges such that every vertex of the graph is incident to at least one edge of the set. Thus, the number of edge covers in the bipartite graph $B$ is given by
\begin{equation}
\#EC(B) = \left| \{S \subseteq E(B) : | \cup_{(x_i,y_j ) {\in} S}\{x_i, y_j\} | = \left| V(B) \right| \} \right|.
 \end{equation}
Since $E(B)=Q$ and $|V(B)|=N$, $\#EC(B)$ can be computed simply as $z_N$. The reduction is thus completed in polynomial time. 
\end{proof}

Since directed path graphs form a subclass of undirected path graphs, the following corollary immediately follows.
\begin{cor}
The \#DS problem for undirected path graphs is \#P-complete.
\end{cor}
\begin{figure}[htbp]
    \makebox[\textwidth]{\includegraphics[scale=1.0, clip, trim=3cm 12.8cm 3cm 2.7cm]{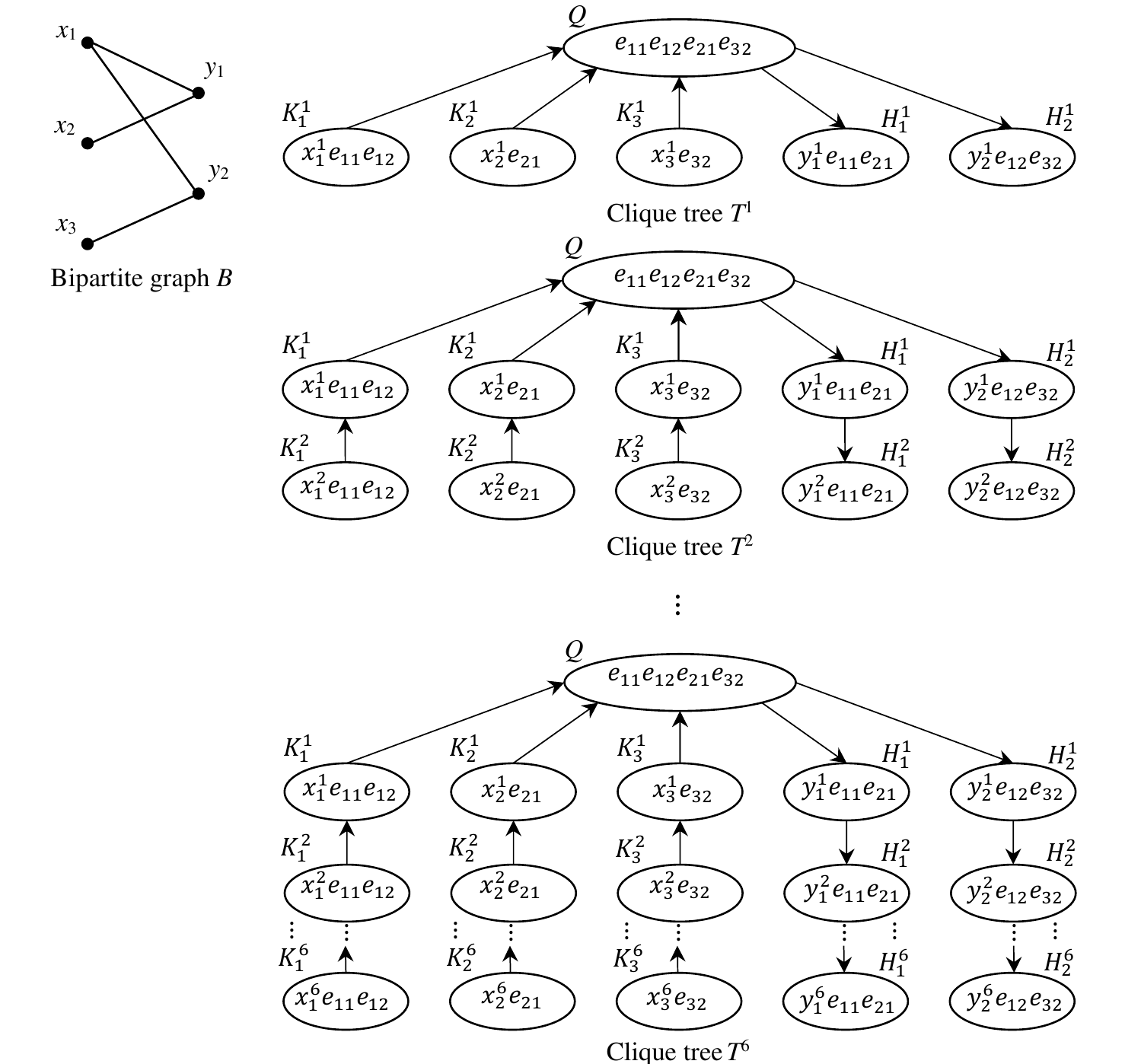}}
    \caption{Bipartite graph $B$ and constructed clique trees $T^r, 1{\leq}r{\leq}6$.}
    \label{fig2}
\end{figure}

\section{Polynomial-time Algorithm for Counting DSs in Rooted Directed Path Graphs}
This section describes a polynomial-time algorithm for solving the \#DS problem for rooted directed path graphs. First, assume that a clique tree $T$ has been constructed for the rooted directed path graph $G$. This construction takes only linear time using an easy modification to the recognition algorithm in \cite{dietz1979linear}. For $v {\in} V(G)$, let $P_v$ be the set of all nodes (cliques) of $T$ that contain vertex $v$. Therefore, $P_v$ forms a subpath in $T$ that is directed away from the root of $T$.
To simplify the description of the algorithm, the clique tree $T$ of a rooted directed path graph $G$ is transformed into an equivalent binary tree $BT$ by the following two steps. 
Step 1: each path $P_v$ in $T$ is extended by appending a leaf node to the ending node of $P_v$. Step 2: each node with an out-degree of $d > 2$ in $T$ are split into $d-1$ nodes with an out-degree of 2; refer to Fig. 3 for details \cite{jansen1997disjoint}. Figure 4 presents an example of the above steps. 
\begin{figure}[htbp]
    \makebox[\textwidth]{\includegraphics[scale=1.0, clip, trim=3cm 22.4cm 3cm 2.7cm]{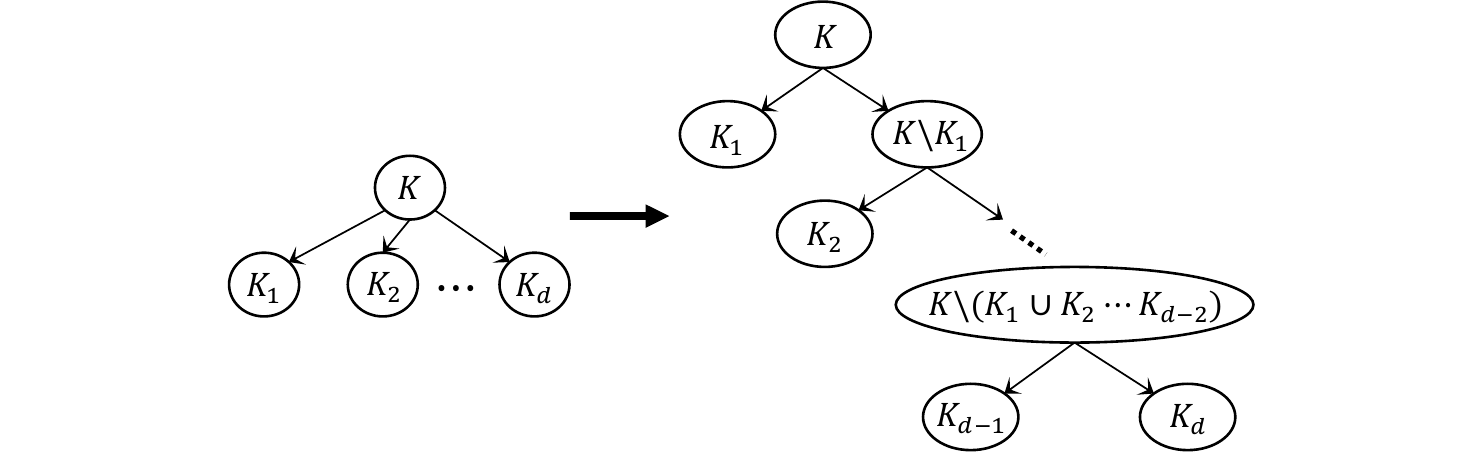}}
    \caption{Transformation of a node with out-degree $d$.}
    \label{fig3}
\end{figure}
\begin{figure}[htbp]
    \makebox[\textwidth]{\includegraphics[scale=1.0, clip, trim=2cm 21.0cm 2cm 2.8cm]{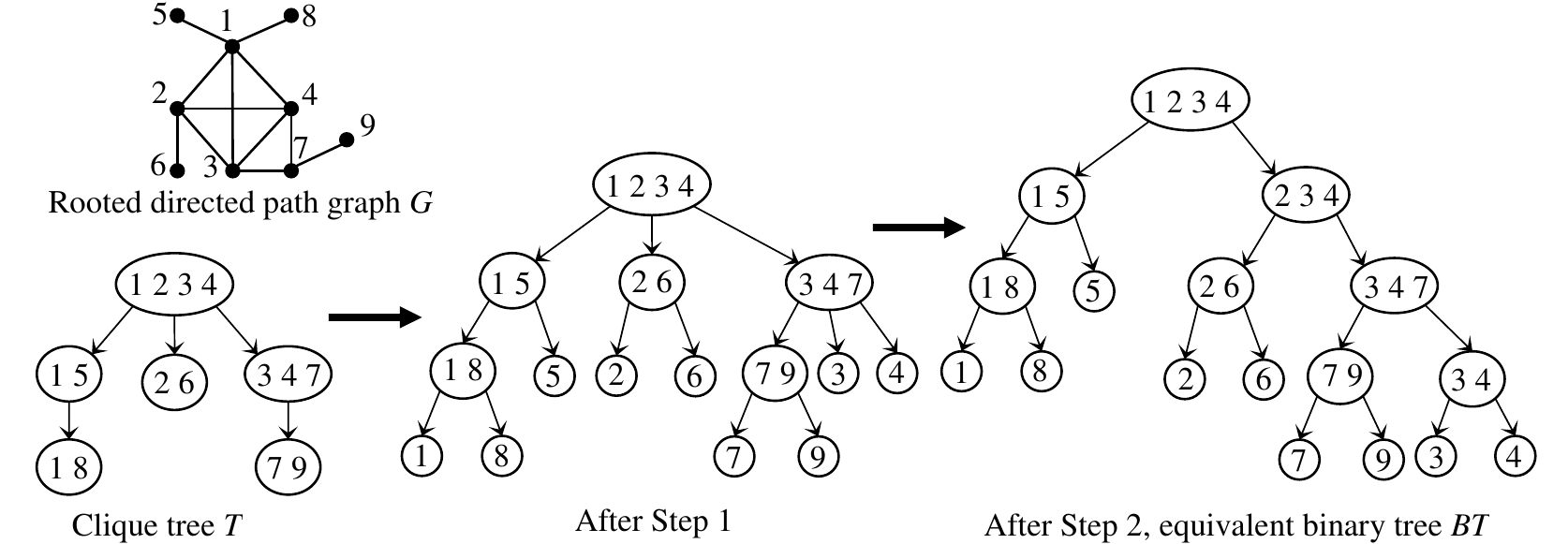}}
    \caption{Example of transforming a clique tree $T$ of a rooted directed path graph $G$ into an equivalent binary tree $BT$.}
    \label{fig4}
\end{figure}
For any node $k$ in $BT$, let $V(k)$ denote the vertices in node $k$ and let $G_k$ be the subgraph of $G$ that is induced by $V(k)$ and the vertices that are contained in all descendants of node $k$ in $BT$. The following property is easily derived.
\begin{pty}
Suppose that node $k$ is an internal node of $BT$ and it has two children $i$ and $j$. Then, 
\begin{equation*}
V(k) \subseteq V(i) \uplus V(j) \text{ and } V(G_k) = V(G_i)  \uplus V(G_j).
\end{equation*}
\end{pty}

Let $S(k)$ be the collection of subsets $S$ of $V(G_k)$ such that $S$ dominates $V(G_k) \setminus V(k)$ in $G_k$. The collection $S(k)$ is partitioned into three disjoint subsets, $A(k), B(k)$, and $C(k)$, determined by whether $S$ contains or dominates the vertices of $V(k)$. For any node $k$ in $BT$, 
\begin{align*}
A(k):=& \{S {\in} S(k) : S \text{ contains some vertex of  }V(k).\}, \\
B(k):=& \{S {\in} S(k) : S \text{ contains no vertex of } V(k) \text{ and some vertices of }V(k) \text{ are not dominated by }S. \}, \text{ and } \\
C(k):=& \{S {\in} S(k) : S \text{ contains no vertex of }V(k) \text{ and all vertices of }V(k) \text{ are dominated by }S.\}. \\
\end{align*}

Notably, all $S {\in} A(k)$ or $S {\in} C(k)$ dominate all vertices of $V(G_k)$ so the number of DSs in $G$ is $|A(r)|+|C(r)|$, where $r$ is the root node of $BT$.
Define a strict partial order $<$ on $V(G)$ such that for any two vertices $u, v {\in} V(G)$, $u < v$ if and only if either the starting node of path $P_u$ is a descendant of the starting node of path $P_v$ in $BT$, or $P_u$ and $P_v$ have the same starting node and $u$ appears before $v$ in an arbitrary predefined sequence of vertices of $V(G)$. According to the strict partial order $<$, $A(k)$ and $B(k)$ are further partitioned into subsets $A(k, v)$ and $B(k, v)$, respectively, for $v {\in} V(k)$, as follows. For any node $k$ in $BT$,
\begin{align*}
A(k, v) :=& \{S {\in} A(k) : v \text{ is the maximum vertex (in the order $<$) of }S \cap V(k). \}, \text{ and } \\
B(k, v) :=& \{S {\in} B(k) : v \text{ is the minimum vertex (in the order $<$) of }V(k) \text{ that is not dominated by }S. \}.
\end{align*}

The following notation will be used in the proof below. Consider two collections of sets $\Omega_1$  and $\Omega_2$, and assume that for each $X_1 {\in} \Omega_1$ and each $X_2 {\in} \Omega_2$, $X_1$ and $X_2$ are disjoint sets. Let $X_1 \uplus X_2$ denote the disjoint union of sets $X_1$ and $X_2$. The operation $\times$  is defined by $\Omega_1 \times \Omega_2=\{X_1 \uplus X_2 : X_1 {\in} \Omega_1 \text{ and } X_2 {\in} \Omega_2\}$. 

\begin{lem}
Suppose that node $k$ is a leaf node of $BT$ and $V(k)=\{v\}$. Then
\begin{equation}
A(k,v)=\large\{ \{v\} \large\}, B(k,v)=\{\emptyset\}, \text{ and } C(k)=\emptyset.
\end{equation}
\end{lem}

\begin{proof}
The lemma follows immediately from the above definitions.
\end{proof}

\begin{lem}
Suppose that node $k$ is an internal node of $BT$ and it has two children $i$ and $j$. For each vertex $v {\in} V(k)$, without loss of generality, assume $v {\in} V(i)$. Then
\begin{equation}
A(k,v)=A(i,v) \times \left \{ \left( \biguplus_{u{\in} V(j)\And u<v}{A(j,u)} \right) \uplus \left( \biguplus_{u{\in} V(j)\cap V(k)}{B(j,u)} \right)\uplus C(j) \right\}.
\end{equation}
\end{lem}

\begin{proof}
Consider a DS $S{\in} A(k,v)$. By Property 1, $S$ can be partitioned into two disjoint subsets $S_i$ and $S_j$, i.e. $S=S_i\uplus S_j$, where $S_i= S\cap V(G_i)$ and $S_j= S\cap V(G_j)$. Since $v{\in} V(i)\cap V(k)$ and $v$ is the maximum vertex of $S\cap V(k)$, $v$ is also the maximum vertex of $S_i\cap V(i)$, and hence $S_i{\in} A(i,v)$. Now consider three possible cases of $S_j$.

\noindent \textbf{Case 1:} $S_j{\in} A(j)$. Notably, not all $S_j{\in} A(j)$ coincide with $S_i{\in} A(i,v)$. Let $u{\in} V(j)$ be the maximum vertex of $S_j\cap V(j)$. If $u{\in} V(k)$, then $u < v$, because $v$ is the maximum vertex of $S\cap V(k)$. If $u{\notin} V(k)$, then node $j$ must be the starting node of path $P_u$, and hence $u < v$. Therefore, $S_j{\in} \biguplus\nolimits_{u{\in} V(j)\And u<v}{A(j,u)}$.

\noindent \textbf{Case 2:} $S_j{\in} B(j)$. Again, not all $S_j{\in} B(j)$ coincide with $S_i{\in} A(i,v)$. Let $u{\in} V(j)$ be the minimum vertex of $V(j)$ that is not dominated by $S_j$. Suppose $u{\notin} V(k)$. Then, $u$ is not dominated by $S$, contradicting the fact that $S{\in} S(k)$ dominates all vertices of $V(G_k)\setminus V(k)$, and hence $u{\in} V(k)$. Therefore, $S_j{\in} \biguplus \nolimits_{u{\in} V(j)\cap V(k)}{B(j,u)}$.

\noindent \textbf{Case 3:} $S_j{\in} C(j)$. Clearly, all $S_j{\in} C(j)$ coincide with $S_i{\in} A(i,v)$.

From Cases 1-3, $S{\in} A(i,v)\times \left\{ \left( \biguplus\nolimits_{u{\in} V(j)\And u<v}{A(j,u)} \right)\uplus \left( \biguplus\nolimits_{u{\in} V(j)\cap V(k)}{B(j,u)} \right)\uplus C(j) \right\}$, and hence $A(k,v)\subseteq A(i,v)\times \left\{ \left( \biguplus\nolimits_{u{\in} V(j)\And u<v}{A(j,u)} \right)\uplus \left( \bigcup\nolimits_{u{\in} V(j)\cap V(k)}{B(j,u)} \right)\uplus C(j) \right\}$.

Conversely, let $S_i' {\in} A(i,v)$ with $v{\in} V(i)\cap V(k)$. Consider the following three cases of $S_j^\prime$.

\noindent \textbf{Case 1:} $S_j'{\in} A(j,u)$ with $u{\in} V(j)$ and $u<v$. Clearly, $S_i'\uplus S_j'{\in} A(k)$ because $S_i'\uplus S_j'$ dominates $V(G_i)\uplus V(G_j)=V(G_k)$ and $S_i'\uplus S_j'$ contains some vertex $v$ of $V(k)$. By the definitions of $A(i,v)$ and $A(j,u)$, $v$ and $u$ are the maximum vertices of $S_i'\cap V(i)$ and $S_j'\cap V(j)$, respectively. Since $u<v$, $v$ is the maximum vertex of $(S_i'\cap V(i))\uplus (S_j'\cap V(j))$. Notably, $(S_i'\cap V(i))\uplus (S_j'\cap V(j))=(S_i'\uplus S_j')\cap (V(i)\uplus V(j))$ and $V(k)\subseteq V(i) \uplus  V(j)$ by Property 1. Hence, $v$ is the maximum vertex of $(S_i'\uplus S_j')\cap V(k)$. Thus, $S_i'\uplus S_j'{\in} A(k,v)$ and $A(i,v)\times \biguplus \nolimits_{u{\in} V(j)\And u<v}{A(j,u)} \subseteq A(k,v)$. 

\noindent \textbf{Case 2:} $S_j'{\in} B(j,u)$ with $u{\in} V(j)\cap V(k)$. By the definition of $B(j,u)$, $u$ is the minimum vertex of $V(j)$ that is not dominated by $S_j'$. Thus, for each vertex $w{\in} V(j)$ that is not dominated by $S_j', w\geq u$, and since $u{\in} V(k), w{\in} V(k)$. Thus, $w$ is dominated by the vertex $v{\in} S_i' $ so $S_i'\uplus S_j'$ dominates all vertices of $V(j)$. Therefore, $S_i'\uplus S_j'$ dominates $V(G_i)\uplus V(G_j)=V(G_k)$. Furthermore, since $v$ is the maximum vertex of $(S_i'\uplus S_j')\cap V(k), S_i'\uplus S_j'{\in} A(k,v)$. Thus, $A(i,v)\times \biguplus \nolimits_{u{\in} V(j)\cap V(k)}{B(j,u)} \subseteq A(k,v)$. 

\noindent \textbf{Case 3:} $S_j'{\in} C(j)$. By the definition of $C(j), S_j'$ dominates $V(G_j)$ so $S_i'\uplus S_j'$ dominates $V(G_i)\uplus V(G_j)=V(G_k)$. In addition, $S_j'$ contains no vertex of $V(j)$ so contains no vertex $V(k)$. Thus, $v$ is the maximum vertex of $(S_i'\uplus S_j')\cap V(k)$. Therefore, $S_i'\uplus S_j' {\in} A(k,v)$ and $A(i,v)\times C(j) \subseteq A(k,v)$.

From the above Cases 1-3, $A(i,v)\times \left\{ \left( \biguplus\nolimits_{u{\in} V(j)\And u<v}{A(j,u)} \right)\uplus \left( \biguplus\nolimits_{u{\in} V(j)\cap V(k)}{B(j,u)} \right)\uplus C(j) \right\} \subseteq A(k,v)$ and the lemma follows. 
\end{proof}

\begin{lem}
Suppose that node $k$ is an internal node of $BT$ and it has two children $i$ and $j$. For each vertex $v{\in} V(k)$, without loss of generality, assume $v{\in} V(i)$. Then
\begin{equation}
B(k,v)=B(i,v)\times \left\{ \left( \biguplus_{u{\in} V(j)\setminus V(k)}{A(j,u)} \right)\uplus \left( \biguplus_{u>v\And u{\in} V(j)\cap V(k)}{B(j,u)} \right)\uplus C(j) \right\}.
\end{equation}
\end{lem}
\begin{proof}
Consider a DS $S{\in} B(k,v)$. Let $S_i= S\cap V(G_i)$ and $S_j= S\cap V(G_j)$. Let $w{\in} V(i)$ be not dominated by $S_i$, and hence not dominated by $S$. Suppose $w < v$. Since $v$ is the minimum vertex of $V(k)$ that is not dominated by $S$, $w{\notin} V(k)$. Thus, there exists some vertex $w{\in} V(G_k)\setminus V(k)$ that is not dominated by $S$, contradicting the fact that $S{\in} S(k)$, and hence $w\geq v$. Thus, $v$ is also the minimum vertex of $V(i)$ that is not dominated by $S_i$. Hence, $S_i{\in} B(i,v)$. Now consider three possible cases of $S_j$.

\noindent \textbf{Case 1:} $S_j{\in} A(j)$. Since $S{\in} B(k,v)$, $S_j$ contains no vertex of $V(k)$. Thus, $S_j{\in} \biguplus \nolimits_{u{\in} V(j)\setminus V(k)}{A(j,u)}$.

\noindent \textbf{Case 2:} $S_j{\in} B(j)$. Let $u{\in} V(j)$ where $u$ is not dominated by $S$, and therefore also not dominated by $S_j$. Suppose $u{\notin} V(k)$. Thus, there exists a vertex $u{\in} V(G_k)\setminus V(k)$ that is not dominated by $S$, contradicting the fact that $S{\in} S(k)$, and hence $u{\in} V(k)$. Additionally, since $v$ is the minimum vertex of $V(k)$ that is not dominated by $S$, $u > v$. Thus, $S_j{\in} \biguplus \nolimits_{u>v\And u{\in} V(j)\cap V(k)}{B(j,u)}$.

\noindent \textbf{Case 3:} $S_j{\in} C(j)$. Clearly, all $S_j{\in} C(j)$ coincide with $S_i{\in} B(i,v)$.

Therefore, from Cases 1-3, $S{\in} B(i,v)\times \left\{ \left( \biguplus \nolimits_{u{\in} V(j)\setminus V(k)}{A(j,u)} \right)\uplus \left( \biguplus \nolimits_{u<v\And u{\in} V(j)\cap V(k)}{B(j,u)} \right)\uplus C(j) \right\}$ and $B(k,v) \subseteq B(i,v)\times \left\{ \left( \biguplus \nolimits_{u{\in} V(j)\setminus V(k)}{A(j,u)} \right)\uplus \left( \biguplus\nolimits_{u>v\And u{\in} V(j)\cap V(k)}{B(j,u)} \right)\uplus C(j) \right\}$.

Conversely, let $S_i'{\in} B(i,v)$ with $v{\in} V(i)\cap V(k)$. Consider the following three cases of $S_j'$.

\noindent \textbf{Case 1:} $S_j'{\in} A(j,u)$ with $u{\in} V(j)\setminus V(k)$. Clearly, $S_i'$ contains no vertex of $V(i)$ so it contains no vertex of $V(k)$. Since $u$ is the maximum vertex of $S_j'\cap V(j)$ and $u{\notin} V(k)$, $S_j'$ contains no vertex of $V(k)$. Thus, $S_i'\uplus S_j'$ contains no vertex of $V(k)$, Thus, $S_i'\uplus S_j'{\in} B(k)$. Additionally, $v$ is the minimum vertex of $V(i$) that is not dominated by $S_i'$. Hence, $v$ is the minimum vertex of $V(i)\cap V(k)$ that is not dominated by $S_i'$. Since $S_j'$ contains $u{\in} V(j)$, all vertices of $V(j)\cap V(k)$ are dominated by $S_j'$. Thus, $v$ is the minimum vertex of $V(k)=(V(i)\cap V(k))\uplus (V(j)\cap V(k))$ that is not dominated by $S_i'\uplus S_j'$. Hence, $S_i'\uplus S_j'{\in} B(k,v)$ and $B(i,v) \times \biguplus \nolimits_{u{\in} V(j)\setminus V(k)}{A(j,u)}\subseteq B(k,v)$.

\noindent \textbf{Case 2:} $S_j'{\in} B(j,u)$ with $u > v$ and $u{\in} V(j)\cap V(k)$. Clearly, $S_i'\uplus S_j'$ contains no vertex of $V(k)$ and $v$ is the minimum vertex of $V(k)$ that is not dominated by $S_i'\uplus S_j'$. Thus, $S_i'\uplus S_j'{\in} B(k,v)$ and $B(i,v)\times \biguplus \nolimits_{u>v\And u{\in} V(j)\cap V(k)}{B(j,u)}\subseteq B(k,v)$.

\noindent \textbf{Case 3:} $S_j'{\in} C(j)$. By the definition of $C(j)$, $S_j$ dominates all vertices of $V(j)$ so it also dominates all vertices of $V(j)\cap V(k)$. In addition, $S_j$ contains no vertex of $V(j)$, and therefore contains no vertex $V(k)$. Thus, $S_i'\uplus S_j'$ contains no vertex of $V(k)$ and $v$ is the minimum vertex of $V(k)$ that is not dominated by $S_i'\uplus S_j'$. Therefore, $S_i'\uplus S_j'{\in} B(k,v)$ and $B(i,v)\times C(j) \subseteq B(k,v)$.

From the above Cases 1-3, $B(i,v)\times \left\{ \left( \biguplus \nolimits_{u{\in} V(j)\setminus V(k)}{A(j,u)} \right)\uplus \left( \biguplus\nolimits_{u>v\And u{\in} V(j)\cap V(k)}{B(j,u)} \right)\uplus C(j) \right\} \subseteq B(k,v)$ and the lemma follows. 
\end{proof}

\begin{lem}
Suppose that node $k$ is an internal node of $BT$ and it has two children $i$ and $j$. Then
\begin{equation}
C(k) = \left\{ \left( \biguplus_{u{\in} V(i)\setminus V(k)}{A(i,u)} \right)\uplus C(i) \right\}\times \left\{ \left( \biguplus_{u{\in} V(j)\setminus V(k)}{A(j,u)} \right)\uplus C(j) \right\}.
\end{equation}
\end{lem}
\begin{proof}

Consider a DS $S{\in} C(k)$. Let $S_i= S\cap V(G_i)$ and $S_j= S\cap V(Gj)$. Consider three possible cases of $S_i$.

\noindent \textbf{Case 1:} $S_i{\in} A(i)$. By the definition of $C(k)$, $S$ contains no vertex of $V(k)$ so $S_i{\in} \biguplus \nolimits_{u{\in} V(i)\setminus V(k)}{A(i,u)}$.

\noindent \textbf{Case 2:} $S_i{\in} B(i)$. Let $w{\in} V(i)$ be a vertex that is not dominated by $S_i$, and hence $w$ is not dominated by $S$. Since $S$ dominates all vertices of $V(k), w{\notin} V(k)$, contradicting the fact that $S{\in} S(k)$ dominates all vertices of $V(G_k)\setminus V(k)$. Therefore, $S_i{\notin} B(i)$.

\noindent \textbf{Case 3:} $S_i{\in} C(i)$. Clearly, all $S_i{\in} C(i)$ coincide with $S{\in} C(k)$.

From the above Cases 1-3, $S_i{\in} \left( \biguplus \nolimits_{u{\in} V(i){\setminus} V(k)}{A(i,u)} \right)\uplus C(i)$. Similarly, $S_j{\in} \left( \biguplus \nolimits_{u{\in} V(j){\setminus} V(k)}{A(j,u)} \right)\uplus C(j)$. Thus, $C(k)\subseteq \left\{ \left( \biguplus\nolimits_{u{\in} V(i){\setminus} V(k)}{A(i,u)} \right)\uplus C(i) \right\}\times \left\{ \left( \biguplus\nolimits_{u{\in} V(j){\setminus} V(k)}{A(j,u)} \right)\uplus C(j) \right\}$.

Conversely, consider the following four cases. 

\noindent \textbf{Case 1:} $S_i'{\in}A(i,v)$ with $v{\in}V(i){\setminus} V(k)$ and $S_j'{\in}A(j,u)$ with $u{\in}V(j){\setminus} V(k)$. Since $v{\in} V(i){\setminus} V(k)$ and $v$ is the maximum vertex of $S_i'\cap V(i), S_i'$ contains no vertex of $V(k)$. Since $S_i'$ dominates all vertex of $V(i)$, $S_i'$ dominates all vertices of $V(i)\cap V(k)$. Similarly, $S_j'$ contains no vertex of $V(k)$ and dominates all vertices of $V(j)\cap V(k)$. Thus, $S_i'\uplus S_j'$ contains no vertex of $V(k)$ and dominates all vertices of $(V(i)\cap V(k))\uplus (V(j)\cap V(k))=V(k)$. Hence, $S_i'\uplus S_j'{\in} C(k)$ and $\left( \biguplus\nolimits_{v{\in} V(i){\setminus} V(k)}{A(i,v)} \right)\times \left( \biguplus\nolimits_{u{\in} V(j){\setminus} V(k)}{A(j,u)} \right) \subseteq C(k)$. 

\noindent \textbf{Case 2:} $S_i'{\in}A(i,v)$ with $v{\in}V(i){\setminus} V(k)$ and $S_j'{\in}C(j)$. By an argument similar to that in Case 1, $S_i'$ contains no vertex of $V(k)$ and all vertices of $V(i)\cap V(k)$ are dominated by $S_i'$. By the definition of $C(j)$, $S_j'$ contains no vertex of $V(j)$ and dominates all vertices of $V(j)$. Hence, $S_j'$ contains no vertex of $V(k)$ and dominates all vertices of $V(j)\cap V(k)$. Therefore, $S_i'\uplus S_j'$ contains no vertex of $V(k)$ and dominates all vertices of $V(k)$. Hence, $S_i'\uplus S_j'{\in} C(k)$ and $\left( \biguplus\nolimits_{v{\in} V(i){\setminus} V(k)}{A(i,v)} \right)\times C(j)\subseteq C(k)$.

\noindent \textbf{Case 3:} $S_i'{\in}C(i)$ and $S_j'{\in}A(j,u)$ with $u{\in}V(j){\setminus} V(k)$. By an argument similar to that in Case 2, $S_i'\uplus S_j'$ contains no vertex of $V(k)$ and dominates all vertices of $V(k)$. Hence, $S_i'\uplus S_j'{\in} C(k)$ and $C(i)\times \left( \biguplus \nolimits_{u{\in} V(j){\setminus} V(k)}{A(j,u)} \right)\subseteq C(k)$.

\noindent \textbf{Case 4:} $S_i'{\in}C(i)$ and $S_j'{\in}C(j)$. By an argument similar to that in Case 2, $S_i'$ contains no vertex of $V(k)$ and dominates all vertices of $V(i)\cap V(k)$. Similarly, $S_j'$ contains no vertex of $V(k)$ and dominates all vertices of $V(j)\cap V(k)$. Thus, $S_i'\uplus S_j'{\in} C(k)$ and $C(i)\times C(j)\subseteq C(k)$.

From the above Cases 1-4, $\left\{ \left( \biguplus \nolimits_{v{\in} V(i){\setminus} V(k)}{A(i,v)} \right)\uplus C(i) \right\}\times \left\{ \left( \biguplus \nolimits_{u{\in} V(j){\setminus} V(k)}{A(j,u)} \right)\uplus C(j) \right\}\subseteq C(k)$ and the lemma follows.
\end{proof}

Let $\Omega_1$ and $\Omega_2$ be two collections of sets. By the definitions of operations $\times$ and $\uplus$, $ |\Omega_1\times \Omega_2| = |\Omega_1|\cdot |\Omega_2 |$ and $|\Omega_1\uplus \Omega_2| = |\Omega_1| + |\Omega_2|$. Thus, based on Lemmas 1-4, the following algorithm can be used to count DSs in a rooted directed path graph.

\begin{algorithm}
\caption{Counting DSs in a rooted directed path graph}
\KwIn{A rooted directed path graph $G$.}
\KwOut{The number of DSs in $G$.}  
\BlankLine    
Construct a rooted directed clique tree $T$; \\
Transform $T$ into an equivalent binary tree $BT$ with root node $r$; \\
\ForEach{node $k{\in}BT$ encountered in the post-order traversal of $BT$}{	
	\uIf{node $k$ is a leaf node with $V(k)=\{v\}$}{
  	     $|A(k, v)|{\leftarrow}1; |B(k, v)|{\leftarrow}1; |C(k)|{\leftarrow}0$;
	}
	\uElse{  \tcp{node $k$ is an internal node of $BT$ with two children}
		\ForEach{vertex $v{\in}V(k)$}{
			node $i\leftarrow$ the child of node $k$ that contains $v$;\\
			node $j\leftarrow$ the child of node $k$ that does not contains $v$; \\
			$|A(k,v)|{\leftarrow}|A(i,v)|\cdot \left( \sum\nolimits_{u{\in}V(j)\And u<v}{|A(j,u)|}+\sum\nolimits_{u{\in}V(j)\cap V(k)}{|B(j,u)|}+|C(j)| \right)$; \\
			$|B(k,v)|{\leftarrow}|B(i,v)|\cdot \left( \sum\nolimits_{u{\in}V(j){\setminus}V(k)}{|A(j,u)}|+\sum\nolimits_{u>v\And u{\in}V(j)\cap V(k)}{|B(j,u)|}+|C(j)| \right)$; \\
			$|C(k)|{\leftarrow}\left( \sum\nolimits_{u{\in}V(i)\setminus V(k)}{|A(i,u)}|+|C(i)| \right)\cdot \left( \sum\nolimits_{u{\in}V(j)\setminus V(k)}{|A(j,u)}|+|C(j)| \right)$; \\
		}
	}
}
\Return{ $\sum\nolimits_{v{\in}V(r)}{|A(r,v)|}+|C(r)|$};
\end{algorithm}

\begin{thm}
For a rooted directed path graph with $n$ vertices, the \#DS problem is solvable in $O(n^3)$ time.
\end{thm}

\begin{proof}
 The correctness of Algorithm 1 follows from Lemmas 1-4. The time complexity of Algorithm 1 is analyzed as follows. Given a rooted directed path graph $G$ with $n$ vertices, the clique tree $T$ can be constructed in linear time by making a simple modification to the recognition algorithm of Dietz et al. \cite{dietz1979linear}. Since $T$ has no more than $n$ nodes, $T$ can be transformed into an equivalent binary tree $BT$ in $O(n^2)$ time. All nodes, except leaf nodes, in $BT$ have an out-degree of 2 and the number of leaf nodes is $n$. Hence, the total number of nodes in $BT$ is $2n-1$. Thus, Algorithm 1 takes $O(n^3)$ time to compute all $|A(k,v)|, |B(k,v)|$, and $|C(k)|$. Consequently, the number of DSs in a rooted directed path graph with $n$ vertices can be computed in $O(n^3)$ time.
\end{proof}

\section{Conclusions}
This paper is the first to determine the complexity of the \#DS problem for two subclasses of chordal graphs - directed path graphs and rooted directed path graphs. The \#DS problem is proved to remain \#P-complete when restricted to directed path graphs but a further restriction to rooted directed path graphs admits a solution in polynomial time.

\bibliographystyle{abbrv}
\bibliography{ms}

\end{document}